\documentclass{eptcs-modified}

\usepackage{amsmath}
\usepackage{amssymb}
\usepackage{url}
\usepackage{amsthm}
\usepackage{graphicx}

\begin{document}

\theoremstyle{definition}
\newtheorem{theorem}{Theorem}
\newtheorem{definition}[theorem]{Definition}
\newtheorem{problem}[theorem]{Problem}
\newtheorem{assumption}[theorem]{Assumption}
\newtheorem{corollary}[theorem]{Corollary}
\newtheorem{proposition}[theorem]{Proposition}
\newtheorem{example}[theorem]{Example}
\newtheorem{lemma}[theorem]{Lemma}
\newtheorem{observation}[theorem]{Observation}
\newtheorem{fact}[theorem]{Fact}
\newtheorem{question}[theorem]{Open Question}
\newtheorem{conjecture}[theorem]{Conjecture}
\newtheorem{addendum}[theorem]{Addendum}
\newtheorem{Wdefn*}[theorem]{Definition(Weihrauch \cite{weihrauchd})}
\newtheorem{Mdefn*}[theorem]{Definition(Moschovakis \cite{moschovakis})}
\newcommand{\uint}{{[0, 1]}}
\newcommand{\Cantor}{{\{0,1\}^\mathbb{N}}}
\newcommand{\name}[1]{\textsc{#1}}
\newcommand{\id}{\textrm{id}}
\newcommand{\dom}{\operatorname{dom}}
\newcommand{\Dom}{\operatorname{Dom}}
\newcommand{\codom}{\operatorname{CDom}}
\newcommand{\Baire}{\mathbb{N}^\mathbb{N}}
\newcommand{\hide}[1]{}
\newcommand{\mto}{\rightrightarrows}
\newcommand{\Sierp}{Sierpi\'nski }
\newcommand{\BC}{\mathcal{B}}
\newenvironment{Wdefn}[1]{\vspace{5pt}\begin{Wdefn*}[#1]}{\end{Wdefn*}\vspace{5pt}}
\newenvironment{Mdefn}[1]{\vspace{5pt}\begin{Mdefn*}[#1]}{\end{Mdefn*}\vspace{5pt}}
\newcommand\tboldsymbol[1]{%
\protect\raisebox{0pt}[0pt][0pt]{%
$\underset{\widetilde{}}{\boldsymbol{#1}}$}\mbox{\hskip 1pt}}

\newcommand{\bcode}{{\rm BC}}
\newcommand{\bcodefun}{\pi}
\newcommand{\bolds}{\tboldsymbol{\Sigma}}
\newcommand{\boldp}{\tboldsymbol{\Pi}}
\newcommand{\boldd}{\tboldsymbol{\Delta}}
\newcommand{\boldg}{\tboldsymbol{\Gamma}}
\newcommand{\pointcl}{\boldsymbol{\Lambda}}

\title{A comparison of concepts from computable analysis and effective descriptive set theory}

\author{
Vassilios Gregoriades
\institute{Department of Mathematics \\ Technische Universit\"at Darmstadt, Germany}
\email{gregoriades@mathematik.tu-darmstadt.de}
\and
Tam\'as Kisp\'eter
\institute{Computer Laboratory\\ University of Cambridge, United Kingdom}
\email{tk407@cam.ac.uk}
\and
Arno Pauly
\institute{Computer Laboratory\\ University of Cambridge, United Kingdom}
\email{Arno.Pauly@cl.cam.ac.uk}
}

\def\titlerunning{Comparing computable analysis and effective DST}
\def\authorrunning{V. Gregoriades, T. Kisp\'eter \& A. Pauly}
\maketitle

\begin{abstract}
Computable analysis and effective descriptive set theory are both concerned with complete metric spaces, functions between them and subsets thereof in an effective setting. The precise relationship of the various definitions used in the two disciplines has so far been neglected, a situation this paper is meant to remedy.

As the role of the Cauchy completion is relevant for both effective approaches to Polish spaces, we consider the interplay of effectivity and completion in some more detail.
\end{abstract}



\section{Introduction}
Both computable analysis (\name{Weihrauch} \cite{weihrauch,weihrauchd}) and effective descriptive set theory (\name{Moschovakis} \cite{moschovakis}) have a notion of computability on (complete, separable) metric spaces as a core concept. Nevertheless, the definitions are prima facie different, and the precise relationship has received little attention so far (contrast e.g.~the well-established connections between \name{Weihrauch}'s and \name{Pour-El} \& \name{Richard}'s approach \cite{pourel} to computable analysis).

The lack of exchange between the two approaches becomes even more regrettable in the light of recent developments that draw on both computable analysis and descriptive set theory: \begin{itemize}
\item The study of Weihrauch reducibility often draws on concepts from descriptive set theory via results that identify various classes of measurable functions as lower cones for Weihrauch reducibility \cite{brattka,paulybrattka,paulydebrecht}. The Weihrauch lattice is used as the setting for a metamathematical investigation of the computable content of mathematical theorems \cite{brattka3,gherardi,paulyincomputabilitynashequilibria}.
\item In fact, Weihrauch reducibility was introduced partly as an analogue to Wadge reducibility for functions (see the original papers by \name{Weihrauch} \cite{weihrauchb,weihrauchc,weihrauchd} and subsequent work by \name{Hertling} \cite{hertling}), and as such, can itself be seen as a subfield of (effective) descriptive set (or rather function) theory.
\item The Quasi-Polish spaces \cite{debrecht6} introduced by \name{de Brecht} allow the generalization of many results from descriptive set theory to a much larger class of spaces (e.g.~\cite{debrecht5,schlicht}), and admit a very natural characterization in terms of computable analysis as those countably based spaces with a total admissible Baire-space representation.
\item Even more so, the suggested \emph{synthetic descriptive set theory} \cite{pauly-descriptive,pauly-descriptive-lics} by the third author and \name{de Brecht} would extend some fundamental results from descriptive set theory even further, to general represented spaces \cite{pauly-synthetic-arxiv}. This could pave the way to apply some very strong results by \name{Kihara} \cite{kihara4} to the long-outstanding questions regarding generalizations of the Jayne-Rogers theorem \cite{jaynerogers,ros,ros2,semmes}.
\end{itemize}

Our goal with the present paper is to facilitate the transfer of results between the two frameworks by pointing out both similarities and differences between definitions. For example, it turns out that the requirements of an effective metric space (as used by \name{Moschovakis}) are strictly stronger than those \name{Weihrauch} imposes on a computable metric space -- however, this is only true for specific metrics, by moving to an equivalent metric, the stronger requirements can always be satisfied. Hence, effective Polish spaces and computable Polish spaces are the same concept.

Besides the fundamental layer of metric spaces, we shall also consider the computability structure on hyperspaces such as all $\Sigma_2$-measurable subsets of some given Polish spaces. While these spaces do not carry a meaningful topology, they can nevertheless be studied as represented spaces. This was done implicitly in \cite{moschovakis}, and more explicitly in \cite{brattka,selivanov5,paulydebrecht} and \cite{pauly-descriptive,pauly-descriptive-lics}.

As a digression, we will consider a more abstract view point on the Cauchy completion to illuminate the different approaches to metric spaces.

\section{Effective Polish Spaces and Computable Polish Spaces}
\label{sec:spaces}
We begin by contrasting the definitions of the fundamental structure on metric spaces used to derive computability notions; Moschovakis defines a \textit{recursively presented metric space} (RPMS) and Weihrauch a \textit{computable metric space} (CMS). Throughout the text, by $\nu_\mathbb{Q} : \mathbb{N} \to \mathbb{Q}$ we denote some standard bijection.
\begin{Mdefn}{3B}
\label{def:rpms}
Suppose $ \mathcal{X} $ is a separable, complete metric space with distance function $ d $. A recursive presentation of $  \mathcal{X} $ is any function $ \mathbf{r} : \mathbb{N} \rightarrow \mathcal{X} $ whose image $  \mathbf{r}[\mathbb{N}] = {r_0,r_1,...} $ is dense in $ \mathcal{X} $ and such that the relations
\[ P^{d,\mathbf{r}}(i,j,k) \Longleftrightarrow d(r_i, r_j) \leq \nu_\mathbb{Q}(k) \]
\[ Q^{d,\mathbf{r}}(i,j,k) \Longleftrightarrow d(r_i, r_j) < \nu_\mathbb{Q}(k) \]
are recursive.\smallskip

A \emph{recursively presented metric space} is a triple $(\mathcal{X}, d, \mathbf{r})$ as above. To every recursively presented metric space $(\mathcal{X}, d, \mathbf{r})$ we assign the \emph{nbhd system} $\{N(\mathcal{X},s) \mid s \in \mathbb{N}\}$, where
\[
N(\mathcal{X},2^{i+1} \cdot 3^{k+1}) = \{x \in \mathcal{X} \mid d(x,r_i) < \nu_{\mathbb{Q}}(k)\},
\]
and $N(\mathcal{X},s)$ is the empty set if $s$ does not have the form $2^{i+1} \cdot 3^{k+1}$.\smallskip

When referring to recursively presented metric space, we usually omit the metric and the recursive presentation, and simply write $\mathcal{X}$.
\end{Mdefn}

\begin{Wdefn}{cf~8.1.2}
\label{def:cms}
We define a computable metric space with its Cauchy representation as follows:
\begin{enumerate}
\item An effective metric space is a tuple $\mathbf{M} = (M, d, (a_n)_{n \in \mathbb{N}})$ such that $(M,d)$ is a metric space and $(a_n)_{n \in\mathbb{N}}$ is a dense sequence in $(M,d)$.
\item The Cauchy representation $ \delta_{\mathbf{M}}  : \subseteq \Baire  M $ associated with the effective metric space $ \mathbf{M} =  (M, d, (a_n)_{n\in \mathbb{N}}) $ is defined by
\[ \delta_{\mathbf{M}}(p) = x \: : \: \Longleftrightarrow  \begin{cases}
      d(a_{p(i)}, a_{p(k)}) \leq 2^{-i} \text{ for } i < k\\
   \text{and } x = \lim\limits_{i\rightarrow \infty}a_{p(i)}
  \end{cases} \]
\item Finally, a computable metric space is an effective metric space such that the following relation (involving a standard numbering $\nu_\mathbb{Q} : \mathbb{N} \to \mathbb{Q}$)
\[ \{(t,u,v,w) \: |\: \nu_{\mathbb{Q}}(t) < d(a_u, a_v) <\nu_{\mathbb{Q}}(w) \} \text{ is r.e.} \]
\end{enumerate}
\end{Wdefn}

Both definitions can only ever apply to separable metric spaces, however, a noticeable difference is Moschovakis' requirement of completeness, which is not demanded by Weihrauch. This is only a superfluous distinction, though:

\begin{observation}
\label{obs:cmscompletion}
If $ \mathbf{M} = (M, d, (a_n)_{n \in\mathbb{N}}) $ is a computable metric space (CMS) with a Cauchy representation then its completion $ \overline{\mathbf{M}} = (\overline{M}, \overline{d}, (a_n)_{n \in\mathbb{N}}) $ (where $ \overline{d} $ is the expanded distance function for the completion, specifically $ \overline{d}|_{M} = d $) is also a CMS.
\end{observation}

A more substantial difference lies in the decidability-requirement of distances between basic points and rational numbers. For Weihrauch's definition, being able to semi-decide $q < d(a_u, a_w)$ and $d(a_u, a_w) < q$ is enough, whereas Moschovakis demands these to be decidable. By identifying\footnote{That this identification actually makes sense follows from the investigation of the class of computable functions between spaces in Section \ref{sec:functions}.} $(a_n)_{n \in\mathbb{N}}$ and $r : \mathbb{N} \to \mathcal{X}$, we immediately find:
\begin{observation}
\label{obsmoschimpweih}
Every recursively presented metric space is a computable metric space.
\end{observation}

The converse fails in general:
\begin{example}
Consider the following CMS: Let the base set be $X = \mathbb{N} \uplus \mathbb{N}$, the dense set also $X$ (with a standard bijection) and the distance function be defined as follows (assuming $ n_i $ is the $i$th element of the first copy of $ \mathbb{N} $, $n'_i$ from the second):

\[ d(n_i, n_j) := |n_i-n_j| \]
\[ d(n_i, n'_i) := 1 + \dfrac{1}{s_i} \]
Where $ s_i $ is the step count of the $ i $th Turing machine started with no arguments if it halts
\[ d(n_i, n'_i) := 1 \]
if it does not. Then, to ensure the validity of the triangle inequality, we set
\[ d(n_j, n'_i) := d(n_i, n'_i) + d(n_i, n_j) \]
\[  d(n'_j, n'_i) := d(n_j, n'_i) + d(n'_j, n_j) \]

This space is a CMS but not an RPMS.
\begin{proof}
To output the upper bound $ d(n_i, n'_i) < 1 + \dfrac{1}{k}\leq q_i $ one only has to simulate $ \varphi_i $, the $ i $th program for $ k $ steps, if it did not halt yet, output $ q_i $, if it did halt it will be a lower bound. We can avoid outputting the exact term for the exact step count in case it halts. Similarly we semidecide the other types of distances.

This will form a CMS (with the representation of eventually constant sequences of points).

Suppose towards a contradiction that $(X,d)$ admits a recursive presentation ${\bf r}: \mathbb{N} \to X$. Since the set ${\bf r}[\mathbb{N}]$ is dense in  $(X,d)$ and the latter space is discrete we have that ${\bf r}$ is surjective. It follows easily that there exists a recursive function $f: \mathbb{N} \times \{0,1\} \to \mathbb{N}$ such that ${\bf r}(f(i,0)) = n_i$ and ${\bf r}(f(i,1)) = n'_i$.

The decidability-requirements now imply in particular that $d({\bf r}(f(i,0)), {\bf r}(f(i,1))) = 1$ is a decidable property -- but by the construction of $d$, this would mean that the Halting problem is decidable, providing the desired contradiction.
 \end{proof}
 \end{example}

We will proceed to find a weaker counterpart to Observation \ref{obsmoschimpweih}. First, note that in a recursively presented metric space we can decide whether $r_n = r_m?$, whereas we cannot decide $a_w = a_u?$ in a computable metric space. It is possible, however, to avoid having duplicate points in the dense sequence even in the latter case. First we shall provide a general criterion for when two dense sequences give rise to homeomorphic computable metric spaces (which presumably is a folklore result):

\begin{lemma}
\label{lemma:idcomputable}
For two CMSs $\mathbf{X} = (M, d, (a_i)_{i \in \mathbb{N}}), \mathbf{X'} = (M,d, (a'_i)_{i \in \mathbb{N}})$ if $a_i$ is uniformly computable in $\mathbf{X'}$ then the
\[ id: \mathbf{X} \rightarrow \mathbf{X'}\]
identity function is computable.
\end{lemma}
\begin{proof}
The assumption means that given $n, k \in \mathbb{N}$, we can compute some $i_{n,k}$ such that $d(a'_{i_{n,k}}, a_n) < 2^{-k}$. Now we are given some $x \in \mathbf{X}$ via some sequence $(a_{n_j})_{j \in \mathbb{N}}$ such that $d(a_{n_j}, x) < 2^{-j}$. Consider the sequence $(a'_{i_{n_{j+1},(j+1)}})_{j \in \mathbb{N}}$. We find that $d(a'_{i_{n_{j+1},(j+1)}}, x) \leq d(a'_{i_{n_{j+1},(j+1)}}, a_{n_{j+1}}) + d(a_{n_{j+1}}, x) \leq 2^{-j-1} + 2^{-j-1} = 2^{-j}$; thus this sequence constitutes a name for $x \in \mathbf{X}'$.
\end{proof}

In general, we shall write $\mathbf{X} \cong \mathbf{Y}$ iff there is a bijection $\lambda : \mathbf{X} \to \mathbf{Y}$ such that $\lambda$ and $\lambda^{-1}$ are computable. In this paper, $\lambda$ will generally be the identity on the underlying sets.

\begin{corollary}\label{TCompIdDenseSets}
For two CMSs $\mathbf{X} = (M, d, (a_i)_{i \in \mathbb{N}}), \mathbf{X'} = (M,d, (a'_i)_{i \in \mathbb{N}})$ if $a_i$ is uniformly computable in $\mathbf{X'}$ and $a'_i$ is uniformly computable in $\mathbf{X}$ then $\mathbf{X} \cong \mathbf{X'}$
\end{corollary}

As a slight detour, we will prove a more general, but ultimately too weak result. We recall from \cite{pauly-synthetic-arxiv} that a represented space is called computably Hausdorff, if inequality is recognizable. Inequality is (computably) recognizable in a represented space $ \textbf{X} $ iff the function $ \neq \ :\ \textbf{X}\times \textbf{X}\rightarrow \mathbb{S} $ is computable, where $ \mathbb{S} $ is the \Sierp space (with underlying set $ \{\bot, \top \} $ and open sets $ \emptyset, \{\top \} $, $ \{\bot, \top \}  $) and $ \neq(x,x) = \bot $ and $ \neq(x,y) = \top $ otherwise.  Note that every computable metric space is computably Hausdorff. We define a multivalued map $\textrm{RemoveDuplicates} : \subseteq \mathcal{C}(\mathbb{N}, \mathbf{X}) \mto \mathcal{C}(\mathbb{N}, \mathbf{X})$ by $\dom(\textrm{RemoveDuplicates}) = \{(x_n)_{n \in \mathbb{N}} \mid \omega = |\{x_n \mid n \in \mathbb{N}\}|\}$ and $(y_n)_{n \in \mathbb{N}} \in \textrm{RemoveDuplicates}((x_n)_{n\in \mathbb{N}})$ iff $\{y_n \mid n \in \mathbb{N}\} = \{x_n \mid n \in \mathbb{N}\}$ and $\forall n \neq m \in \mathbb{N} \ . \ y_n \neq y_m$. In words, $\textrm{RemoveDuplicates}$ takes a sequence with infinite range, and produces a sequence with the same range but without duplicates.

\begin{proposition}
\label{prop:removeduplicates}
 Let $\mathbf{X}$ be computably Hausdorff. Then $\textrm{RemoveDuplicates} : \subseteq \mathcal{C}(\mathbb{N}, \mathbf{X}) \mto\mathcal{C}(\mathbb{N}, \mathbf{X})$ is computable.
\begin{proof}
 Given a sequence $(x_n)_{n \in \mathbb{N}}$ in a computable Hausdorff space, we can compute $\{n \in \mathbb{N} \mid \forall i < n \ x_i \neq x_n\} \in \mathcal{O}(\mathbb{N})$, i.e.~as a recursively enumerable set (relative to the sequence). By assumption on the range of the sequence, this set is infinite. It is a basic result from recursion theory that any infinite recursively enumerable set is the range of an injective computable function, and this holds uniformly. Let $\lambda$ be such a function. Then $y_n = x_{\lambda(n)}$ satisfies the criteria for the output.
 \end{proof}
 \end{proposition}

The combination of Lemma \ref{lemma:idcomputable} and Proposition \ref{prop:removeduplicates} allows us to conclude that for any infinite computable metric space $\mathbf{X}$, there is a computable metric space $\mathbf{X}'$ with the same underlying set and metric, and a repetition-free dense sequence such that $\id : \mathbf{X}' \to \mathbf{X}$ is computable -- but we cannot guarantee computability of $\id : \mathbf{X} \to \mathbf{X}'$ thus. Consequently, we shall employ a more complicated construction:

\begin{theorem}\label{TCompIdentCMS}
For any infinite CMS $ \mathbf{X} = (M, d, (a_i)_{i \in \mathbb{N}})$, there is a repetition-free sequence $(a'_i)_{i \in \mathbb{N}}$ such that $\mathbf{X}' = (M, d, (a'_i)_{i \in \mathbb{N}})$ is a CMS with $\mathbf{X} \cong \mathbf{X}'$.
\end{theorem}
\begin{proof}
We will first describe an algorithm obtaining the sequence $(a'_i)_{i \in \mathbb{N}}$ from the original sequence $(a_i)_{i \in \mathbb{N}}$.
\begin{enumerate}
\item At any stage, let $ A' $ be the finite prefix sequence of the $(a'_i)_{i \in \mathbb{N}}$ emitted so far. We also keep track of a precision parameter $n$, starting with $n := 1$.
\item In the first stage, we emit $a_0$ into $ A' $ (i.e.~we set $a'_0 := a_0$)
\item Do the following iteration:
\begin{enumerate}
\item Take the next element from $(a_i)_{i \in \mathbb{N}}$ and place it in an auxiliary set $ B $, increment $ n $
\item For all elements  $b \in B $, we can compute the number $ \min_b = \min\{d(a,b) \; | \; a \in A'\} \in \mathbb{R}$ where $ A' $ is the finite sequence of $ a'_i $s emitted so far.
\item For each $ \min_b $, check (non-deterministically) in parallel: if $ \min_b < 2^{-n} $ skip $ b $, if $ \min_b > 2^{-n-1} $ emit $ b $, remove $ b $ from $ B $, emit $ b $ as an $ a'_i $  (thus also suffix it on $ A' $),  repeat.
\item If all elements in $ B $ were skipped, repeat.
\end{enumerate}
\end{enumerate}

The parallel test in $3 (c)$ is a common trick in computable analysis. The relations by themselves are not decidable, but as at least one of them has to be true, we can wait until we recognize a true proposition. If there are multiple $ b $ such that $ \min_b$ lies between $ 2^{-n-1} $ and $ 2^{-n} $, then the choice is non-deterministic in the high level view of real numbers as inputs. If all codings and implementations are fixed, then the choice here is determined, too, though.

First, we shall argue that $(a'_i)_{i \in \mathbb{N}}$ is dense and repetition free. If $(a'_i)_{i \in \mathbb{N}}$ were not dense, then there would be some $m, k \in \mathbb{N}$ such that $\forall i \in \mathbb{N} \ d(a_m, a'_i) > 2^{-k}$. However, then once $m$ has been placed into $B$ and $n$ incremented beyond $k + 1$, $m$ would have been chosen for $A'$ -- contradiction. The sequence $(a'_i)_{i \in \mathbb{N}}$ cannot have repetitions, because a duplicate element could never satisfy the test in $3 (c)$.

In remains to show that $(a'_i)_{i \in \mathbb{N}}$ is computable in $(a_i)_{i \in\mathbb{N}}$ and vice versa. From Corollary~\ref{TCompIdDenseSets} we would then know that $ (M, d, (a_i)_{i \in\mathbb{N}}) \cong (M,d,(a'_i)_{i \in \mathbb{N}}) $.

By construction $(a'_i)_{i \in \mathbb{N}}$ is computable in $(M, d, A)$: Given $i \in \mathbb{N}$, just follow the construction above in order to identify which $a_j$ is the $i$-th element to be put into $A'$, then we have $a'_i = a_j$.

Now to prove that $(a_i)_{i \in\mathbb{N}}$ is computable in $ (M, d, (a'_i)_{i \in \mathbb{N}}) $; i.e.~that given some $i \in \mathbb{N}$ we can compute a sequence $(n_j)_{j \in \mathbb{N}}$ such that $d(a'_{n_j}, a_i) < 2^{-j}$. For this, we inspect the algorithm above beginning from the point when $a_i$ is put into $B$. If $a_i$ is moved into $A'$ as the $k$-th element to enter $A'$, then $d(a'_k, a_i) = 0$, and we can continue the sequence $(n_j)_{j \in \mathbb{N}}$ as the constant sequence $k$. If $a_i$ is not moved into $A'$ in the $j$-th round, then this is due to $ \min_{a_i} < 2^{-j} $, and there must be some $l$ such that $a'_l$ witnesses this distance, i.e.~$d(a_i, a'_l) < 2^{-j}$. Thus, continuing the sequence with $n_j := l$ works.
\end{proof}

In \cite{gregoriades} it is proved that for every recursively presented metric space $(X,d)$ there exists a recursive real $0<\alpha < 1$ such that the metric $\alpha \cdot d$ takes values in $\mathbb{R} \setminus \mathbb{Q} \cup \{0\}$ on the dense sequence. This idea combined with Theorem \ref{TCompIdentCMS} gives the following result.

\begin{theorem}\label{theorem:from_cms_to_rps}
For every CMS $\mathbf{X} = (M, d, (a_i)_{i \in \mathbb{N}})$ there is a CMS $\mathbf{X}' = (M, \alpha d, (a'_i)_{i \in \mathbb{N}})$ with a computable real $\alpha \leq 1$ such that $\mathbf{X} \cong \mathbf{X}'$, and $\mathbf{X}'$ satisfies the criteria for a recursively presented metric space.
\begin{proof}
If $\mathbf{X}$ is finite, the result is straight-forward. If $\mathbf{X}$ is infinite, we may assume by Theorem \ref{TCompIdentCMS} that $(a_i)_{i \in \mathbb{N}}$ is repetition-free. Let the following be computable bijections:
 \begin{enumerate}
 \item $\langle, \rangle_{-\Delta} : \left (\mathbb{N} \times \mathbb{N} \setminus \{(n, n) \mid n \in \mathbb{N}\} \right ) \to \mathbb{N}$
  \item $\nu_\mathbb{Q}^+ : \mathbb{N} \to \{q \in \mathbb{Q} \mid q > 0\}$
 \item $\langle, \rangle : \mathbb{N} \times \mathbb{N} \to \mathbb{N}$
 \end{enumerate}
 Then consider the computable sequence defined via $D_{\langle k, \langle i, j\rangle_{-\Delta}\rangle} = \frac{\nu_\mathbb{Q}^+(k)}{d(a_i, a_j)}$. We can diagonalize to find a computable real number $ \alpha $ not in  $(D_n)_{n \in \mathbb{N}}$ with $0.5 \leq \alpha \leq 1$. By choice of $\alpha$, we find $\alpha d(a_i, a_j) \notin \mathbb{Q}$ for $i \neq j$, hence, the problematic case in the requirements for a recursively presented metric space becomes irrelevant. To compute the identity $\id : \mathbf{X} \to \mathbf{X}'$, one just needs to map a fast Cauchy sequence $(x_i)_{i \in \mathbb{N}}$ to $(x_{i+1})_{i \in \mathbb{N}}$ (as $\alpha \geq 0.5$), the identity in the other direction does not require any changes at all.
\end{proof}
\end{theorem}
\section{The induced computability structures}
\label{sec:functions}
Having proved that computable- and recursively presented metric spaces are inter-connected through a computable rescaling of the metric, it is natural to compare some of the basic objects derived from them. There are three fundamental types of objects in recursively presented metric spaces: recursive- sets, functions and points, all of which have the corresponding analogue in computable metric spaces. We will see that these concepts do coincide.

A comparison of the computability structure induced by recursive presentations and computable metric spaces respectively is more illuminating in the framework of represented spaces. We recall some notions from \cite{pauly-synthetic-arxiv}, and then prove some basic facts about them -- these results are known and included here for completeness. A \emph{represented space} is a pair $\mathbf{X} = (X, \delta_X)$ of a set $X$ and a partial surjection $\delta_X : \subseteq \Baire \to X$. A multi-valued function between represented spaces is a multi-valued function between the underlying sets. For $f : \subseteq \mathbf{X} \mto \mathbf{Y}$ and $F : \subseteq \Baire \to \Baire$, we call $F$ a realizer of $f$ (notation $F \vdash f$), iff $\delta_Y(F(p)) \in f(\delta_X(p))$ for all $p \in \dom(f\delta_X)$. A map between represented spaces is called computable (continuous), iff it has a computable (continuous) realizer. Similarly, we call a point $x \in \mathbf{X}$ computable, iff there is some computable $p \in \Baire$ with $\delta_\mathbf{X}(p) = x$. Any computable metric space induces a represented space via its Cauchy representation, and a function between computable metric spaces is called computable, iff it is computable between the induced representations. Note that the realizer-induced notion of continuity coincides with ordinary metric continuity is a basic fact about admissible representation \cite{weihrauchd}.

The category of represented spaces is cartesian closed, meaning we have access to a general function space construction as follows: Given two represented spaces $\mathbf{X}$, $\mathbf{Y}$ we obtain a third represented space $\mathcal{C}(\mathbf{X}, \mathbf{Y})$ of functions from $X$ to $Y$ by letting $0^n1p$ be a $[\delta_X \to \delta_Y]$-name for $f$, if the $n$-th Turing machine equipped with the oracle $p$ computes a realizer for $f$. As a consequence of the UTM theorem, $\mathcal{C}(-, -)$ is the exponential in the category of continuous maps between represented spaces, and the evaluation map is even computable.

Still drawing from \cite{pauly-synthetic-arxiv}, we consider the \Sierp space $\mathbb{S}$, which allows us to formalize semi-decidability. An explicit representation for this space is $\delta_\mathbb{S} : \Baire \to \mathbb{S}$ where $\delta_\mathbb{S}(0^\mathbb{N}) = \bot$ and $\delta_\mathbb{S}(p) = \top$ for $p \neq 0^\mathbb{N}$. The computable functions $f : \mathbb{N} \to \mathbb{S}$ are exactly those where $f^{-1}(\{\top\})$ is recursively enumerable (and thus $f^{-1}(\{\bot\})$ co-recursively enumerable). In general, for any represented space $\mathbf{X}$ we obtain two spaces of subsets of $\mathbf{X}$; the space of open sets $\mathcal{O}(\mathbf{X})$ by identifying $f \in \mathcal{C}(\mathbf{X}, \mathbb{S})$ with $f^{-1}(\{\top\})$, and the space of closed sets $\mathcal{A}(\mathbf{X})$ by identifying $f \in \mathcal{C}(\mathbf{X}, \mathbb{S})$ with $f^{-1}(\{\bot\})$. In particular, the computable elements of $\mathcal{O}(\mathbb{N})$ are precisely the recursively enumerable sets. An explicit representation for $\mathcal{O}(\mathbb{N})$ is found in $\delta_{rng} : \Baire \to \mathcal{O}(\mathbb{N})$ defined via $n \in \delta_{rng}(p)$ iff $\exists i \in \mathbb{N} \ p(i) = n + 1$.

The focus of computable analysis has traditionally been on the computable admissible spaces. Following \name{Schr\"oder} \cite{schroder5} we call a space $\mathbf{X}$ computably admissible, iff the canonic map $\kappa : \mathbf{X} \to \mathcal{O}(\mathcal{O}(\mathbf{X}))$ has a computable inverse. This essentially means that a point can be effectively recovered from its neighborhood filter. The computably admissible spaces are those represented spaces that correspond to topological spaces.

We proceed to introduce the notion of an effective countable base. The effectively countably based computable admissible spaces are exactly the computable topological spaces studied by \name{Weihrauch} (e.g.~\cite[Definition 3.1]{rettinger2}). The countably based admissible spaces that admit a total Baire space representation are the Quasi-Polish spaces introduced by \name{de Brecht} \cite{debrecht6}.

\begin{definition}
An effective countable base for $\mathbf{X}$ is a computable sequence $(U_i)_{i \in \mathbb{N}} \in \mathcal{C}(\mathbb{N}, \mathcal{O}(\mathbf{X}))$ such that the multivalued partial map $\textrm{Base} :\subseteq \mathbf{X} \times \mathcal{O}(\mathbf{X}) \mto \mathbb{N}$ is computable. Here $\dom(\textrm{Base}) = \{(x, U) \mid x \in U\}$ and $n \in \textrm{Base}(x, U)$ iff $x \in U_n \subseteq U$. Note that the requirement on $ \textrm{Base} $ also gives that $ (U_i)_{i \in \mathbb{N}} $ forms a basis of $ \mathbf{X} $.
\end{definition}

\begin{proposition}[(\footnote{As mentioned in the introduction to this section, this result is folklore. It has appeared e.g.~as \cite[Theorem 2.3]{korovina}.})]
\label{prop:cmshasbase}
Let $\mathbf{X} = (M, d, (a_i)_{i \in \mathbb{N}})$ be a CMS. Then $B_{\langle i, j\rangle} = \{x \in M \mid d(x, a_i) < 2^{-j}\}$ provides an effective countable base for $\mathbf{X}$.
\begin{proof}
We start to prove that this is a computable sequence. By the definition of $\mathcal{O}$, it suffices to show that given $x$, $i$, $j$ we can recognize $d(x, a_i) < 2^{-j}$. Let $\delta_{M}(p) = x$, i.e.~$\forall k \ d(x, a_{p(k)}) < 2^{-k}$. Now $d(x, a_i) < 2^{-j}$ iff $\exists k \in \mathbb{N} \ d(a_{p(k)}, a_i) < 2^{-j} - 2^{-k}$. By the conditions on a CMS, the property is r.e., and existential quantification over an r.e.~property still produces an r.e.~property.

Next, we need to argue that $\textrm{Base}$ is computable. Given some $x \in M$ and some open set $U \in \mathcal{O}(\mathbf{X})$ with $x \in U$, we do know by definition of $\mathcal{O}(\mathbf{X})$ that $x \in U$ will be recognized at some finite stage. Moreover, we can simulate the computation until this happens. At this point, only some finite prefix of the $\delta_M$-name $p$ of $x$ has been read, say of length $N$. But then we must have $x \in \bigcap_{k \leq N} B_{\langle p(k), k\rangle} \subseteq U$. It is easy to verify that we can identify a particular ball inside the intersection still containing $x$.
\end{proof}
\end{proposition}

We now have the ingredients to give a more specific characterization of both $\mathcal{C}(\mathbf{X}, \mathbf{Y})$ and $\mathcal{O}(\mathbf{X})$ for countably based spaces $\mathbf{X}$, $\mathbf{Y}$ and computably admissible $\mathbf{Y}$.

\begin{proposition}
\label{prop:countablybasedopens}
Let $\mathbf{X} = (X, \delta_\mathbf{X})$ have an effective countable base $(U_i)_{i \in \mathbb{N}}$ and let $(p_n)_{n \in \mathbb{N}}$ be a computable sequence that is dense in $\dom(\delta_\mathbf{X})$. Then the map $\bigcup : \mathcal{O}(\mathbb{N}) \to \mathcal{O}(\mathbf{X})$ defined via $\bigcup (S) = \bigcup_{i \in S} U_i$ is computable and has a computable multivalued inverse.
\begin{proof}
That the map is computable follows from \cite[Proposition 4.2(4), Proposition 3.3(4)]{pauly-synthetic-arxiv}. For the inverse, fix some computable realizer of $\textrm{Base}$. Given some $U \in \mathcal{O}(\mathbf{X})$, test for any $n \in\mathbb{N}$ if $\delta_\mathbf{X}(p_n) \in U$. If this is confirmed, compute $m_n := \textrm{Base}(\delta_\mathbf{X}(p_n),U)$ and list it in $\bigcup^{-1}(U)$.

We will now argue that $\bigcup \bigcup^{-1}(U) = U$ with the algorithm described above. If $m \in \bigcup^{-1}(U)$, then by construction $U_m \subseteq U$, hence $\bigcup \bigcup^{-1}(U) \subseteq U$. On the other hand, let $x = \delta_\mathbf{X}(q) \in U$. The realizer for $\textrm{Base}$ will choose some $m_q$ on input $q$, $U$. As this happens after some finite time, there is some $p_{i_q}$ so close to $x$ that the realizer works in exactly the same way\footnote{For this, it is important to fix one realizer of $\textrm{Base}$ and to use the same name of $U$ for all calls.}. This ensures that $m_q$ is listed in $\bigcup^{-1}(U)$, thus $x \in \bigcup \bigcup^{-1}(U)$.
\end{proof}
\end{proposition}

\begin{proposition}
\label{prop:countablybasesubsetscott}
Let $\mathbf{X} = (X, \delta_\mathbf{X})$ be computably admissible and have an effective countable base $(U_i)_{i \in \mathbb{N}}$, and let $(p_n)_{n \in \mathbb{N}}$ be a computable sequence that is dense in $\dom(\delta_\mathbf{X})$. Then $x \mapsto \{i \mid x \in U_i\} : \mathbf{X} \to \mathcal{O}(\mathbb{N})$ is a computable embedding.
\begin{proof}
That the map is computable is straight-forward. For the inverse, we shall first argue that $\{i \in \mathbb{N} \mid x \in U_i\} \mapsto \{U \in \mathcal{O}(\mathbf{X}) \mid x \in U\} : \mathcal{O}(\mathbb{N}) \to \mathcal{O}(\mathcal{O}(\mathbf{X}))$ is computable. By type-conversion, this is equivalent to $(\{i \in \mathbb{N} \mid x \in U_i\}, U) \mapsto (x \in U?) : \mathcal{O}(\mathbb{N}) \times \mathcal{O}(\mathbf{X}) \to \mathbb{S}$. Here we understand $(x \in U?) = \top$ if $x \in U$ and $(x \in U?) = \bot$ if $x \notin U$. By employing Proposition \ref{prop:countablybasedopens}, this follows from $\mathalpha{\in} : \mathcal{O}(\mathbb{N}) \times \mathcal{O}(\mathcal{O}(\mathbb{N})) \to \mathbb{S}$ being computable.

Finally, we can compute $x$ from $\{U \in \mathcal{O}(\mathbf{X}) \mid x \in U\}$ as $\mathbf{X}$ is computably admissible.
\end{proof}
\end{proposition}

As a consequence of Proposition \ref{prop:countablybasedopens}, we see that for countably based spaces $\mathbf{X}$, we may conceive of open sets being given by enumeration of basic open sets exhausting them. For computable metric spaces in particular, an open set is given by an enumeration of open balls with basic points as centers and radii of the form $2^{-i}$ (or equivalently, rational radii):

\begin{Wdefn}{cf~4.1.2}
Given a computable metric space $\mathbf{X}$, we define a numbering $\text{I}$ for the open balls with basic centers and radii of the form $2^{-i}$ via $\text{I}(\langle n, k\rangle) = B(a_n, 2^{-k})$. For convenience, we shall assume that $\text{I}(0) = \emptyset$.
\end{Wdefn}

\begin{Wdefn}{5.1.15.4}
Given a computable metric space $\mathbf{X}$, we define the representation $ \theta_<^{en} : \Baire \to \mathcal{O}(\mathbf{X})$ by
\[ \theta_<^{en}(p) \, :=\,  \bigcup_{n \in \mathbb{N}} \text{I}(p(n))\]
which is intuitively a name consisting of the descriptions of open balls that exhaust the particular set (but not necessarily all of them).\smallskip

An open $V \subseteq \mathbf{X}$ is \emph{computably open} if $V = \theta_<^{en}(p)$ for some recursive $p \in \Baire$.
\end{Wdefn}

The analogous notion in effective descriptive set theory is the following.

\begin{Mdefn}{1B.1}
Given a recursively presented metric space $\mathcal{X}$, a pointset $V \subseteq \mathcal{X} $ is semirecursive (or else $ \Sigma_1^0 $) if
\[ V = \bigcup_n N(\mathcal{X}, \varepsilon(n)) \]
with some recursive $\varepsilon \: : \: \omega \rightarrow \omega.$
\end{Mdefn}

The following lemma is simple but useful tool.

\begin{lemma}
\label{lemma:auxiliary_to_comparison}
Suppose that $(\mathcal{X},d,\mathbf{r})$ is a recursively presented metric space, which by Obesrvation \ref{obsmoschimpweih} is a computable metric space. Then there exist computable functions
\[
\tau: \mathbb{N} \to \mathbb{N}, \quad \sigma: \mathbb{N} \to \mathbb{N}^2
\]
such that
\[
\text{I}(w) = N(\mathcal{X}(\tau(w))), \quad \text{and} \quad N(\mathcal{X},s) = \bigcup_{n}\text{I}(\sigma(s,n))
\]
for all $w,n$.
\end{lemma}

\begin{proof}
The existence of such a function $\tau$ follows easily since in the definitions we use computable encoding. Regarding $\sigma$, we first claim that
\[
N(\mathcal{X},s) = \cup_{j: \ r_j \in N(\mathcal{X},s)} \cup_{m: \ 2^{-m} < \nu_{\mathbb{Q}(k)} - d(r_i,r_j)} B(r_j,2^{-m}),
\]
where $s = 2^{i+1} \cdot 3^{k+1}$.

To see the latter, assume first that $x \in N(\mathcal{X},s)$, where $s = 2^{i+1} \cdot 3^{k+1}$. Then $d(x,r_i) < \nu_{\mathbb{Q}}(k)$. We choose $m$ large enough such that $2^{-(m+1)} < \nu_{\mathbb{Q}}(k) - d(x,r_i)$, from which it follows that $2^{-m} <  \nu_{\mathbb{Q}}(k) - d(x,r_i) - 2^{-m}$. Now we consider some $r_j \in B(x, 2^{-m})$. Clearly $x$ belongs to $B(r_j,2^{-m})$, and
\[
d(r_i,r_j) \leq d(r_i,x) + d(x,r_j) < d(r_i,x) + 2^{-m} < d(r_i,x) + \nu_{\mathbb{Q}}(k) - d(x,r_i) - 2^{-m} = \nu_{\mathbb{Q}}(k)  - 2^{-m},
\]
hence $2^{-m} < \nu_{\mathbb{Q}(k)} - d(r_i,r_j)$. This also implies that $d(r_j,r_i) < \nu_{\mathbb{Q}}(k)$ and so $r_j \in N(\mathcal{X},s)$. Hence $(j,m)$ is a suitable pair of naturals such that $x \in B(r_j,2^{-m})$. Conversely if $j,m$ are such that $2^{-m} < \nu_{\mathbb{Q}(k)} - d(r_i,r_j)$, and $x$ is a member of $B(r_j,2^{-m})$, then we have that
\[
d(x,r_i) \leq d(x,r_j) + d(r_j, r_i) < 2^{-m} +d(r_j, r_i) < \nu_{\mathbb{Q}}(k),
\]
and we have proved the preceding equality. Now we consider some recursive function $w^\ast: \mathbb{N}^2 \to \mathbb{N}$ such that $B(r_j,2^{-m}) = \text{I}(w^\ast(j,m))$ and we define
\[
\sigma(2^{i+1} \cdot 3^{k+1}, 2^{j+1} \cdot 3^{m+1}) =
\begin{cases}
w^\ast(j,m),& \ \text{if} \ 2^{-m} < \nu_{\mathbb{Q}(k)} - d(r_i,r_j)\\
0, & \ \text{else}.
\end{cases}
\]
We also let $\sigma(s,n)$ be $0$ if the naturals $s,n$ do not have the form above.
\end{proof}

We are now ready to compare the notions of computably-open and semirecursive set.

\begin{theorem}
\label{theorem:computably_open_sets_the_same}
Suppose that $\mathcal{X}$ is a recursively presented metric space, $\mathbf{X} = (M,d,(a_i)_{i \in \mathbb{N}})$ is a computable metric space, and $\mathbf{X}'\cong \mathbf{X}$ is as in Theorem \ref{theorem:from_cms_to_rps}. Then:
\begin{enumerate}
\item For every $V \subseteq \mathcal{X}$,\vspace*{-4mm}
\[
V \ \text{is semirecursive} \iff V \ \text{is computably open},\vspace*{-2mm}
\]
{(}recall from Observation \ref{obsmoschimpweih} that $\mathcal{X}$ is also a computable metric space\textup{)}.
\item  For every $U \subseteq M$,\vspace*{-2mm}
\[
U \ \text{is computably open in $\mathbf{X}$ (equivalently in $\mathbf{X}'$)}  \iff U \ \text{is semirecursive in} \ \mathbf{X}',\vspace*{-2mm}
\]
(recall that $\mathbf{X}'$ is recursively presented).\smallskip
\end{enumerate}
In particular, the family of all semirecursive subsets of a recursively presented metric space is also the family of all computably open subsets of the latter space; and the family of all computably open subsets of a computable metric space is the family of all semirecursive subsets of a recursive presented metric space, which is $\cong$-equivalent to the original one.
\end{theorem}

\begin{proof}

The second assertion follows from the first one and the fact that the metric space $\mathbf{X}'$ is recursively presented, so let us prove the first assertion. Let $V \subseteq \mathcal{X}$ and assume that $V$ is semirecursive. Then $V = \cup_{m} N(\mathcal{X},\varepsilon(m))$ for some recursive $\varepsilon$. From Lemma \ref{lemma:auxiliary_to_comparison} we have that
\[
V = \bigcup_{m} N(\mathcal{X},\varepsilon(m)) = \bigcup_{m,n}\text{I}(\sigma(\varepsilon(m),n))
\]
and $V$ is computably open from the closure properties of the latter class of sets, cf.~\cite[Proposition 6 (4)]{pauly-synthetic-arxiv}. The converse follows again from Lemma \ref{lemma:auxiliary_to_comparison} by using the function $\tau$ and the closure properties of semirecursive sets, cf.~\cite{moschovakis} 3C.1 (closure under $\exists^\omega$).
\end{proof}

We now shift our attention to computable/recursive functions.

\begin{proposition}
For two represented spaces $\mathbf{X}$, $\mathbf{Y}$ the map $f \mapsto \{(x, U) \mid f(x) \in U\} : \mathcal{C}(\mathbf{X},\mathbf{Y}) \to \mathcal{O}(\mathbf{X} \times \mathcal{O}(\mathbf{Y}))$ is computable. If $\mathbf{Y}$ is computably admissible, then this map admits a computable inverse.
\begin{proof}
That $f \mapsto \{(x, U) \mid f(x) \in U\}$ is computable follows by combining computability of $f \mapsto f^{-1}: \mathcal{C}(\mathbf{X},\mathbf{Y})\to\mathcal{C}(\mathcal{O}(\mathbf{Y}),\mathcal{O}(\mathbf{X}))$ and computability of $\mathalpha{\in} : \mathbf{X} \times \mathcal{O}(\mathbf{X}) \to \mathbb{S}$ and using type conversion.

For the inverse direction, recall that for computably admissible $\mathbf{Y}$ the map $\{U \in \mathcal{O}(\mathbf{Y}) \mid y \in U\} \mapsto y :\subseteq \mathcal{O}(\mathcal{O}(\mathbf{Y})) \to \mathbf{Y}$ is computable. By computability of $\operatorname{Cut} : \mathbf{X} \times \mathcal{O}(\mathbf{X} \times \mathcal{O}(\mathbf{Y})) \to \mathcal{O}(\mathcal{O}(\mathbf{Y}))$ defined via $\operatorname{Cut}(x,V) = \{U \mid (x,U) \in V\}$ and composition, we find that $(x_0,  \{(x, U) \mid f(x) \in U\}) \mapsto f(x_0) : \subseteq \mathbf{X} \times \mathcal{O}(\mathbf{X} \times \mathcal{O}(\mathbf{Y})) \to \mathbf{Y}$ is computable. Currying produces the claim.
\end{proof}
\end{proposition}

\begin{corollary}
\label{corollary:computable_function}
Let $\mathbf{Y}$ be computably admissible. Then $f :\mathbf{X} \to \mathbf{Y}$ is computable iff $\{(x, U) \mid f(x) \in U\} \in \mathcal{O}(\mathbf{X} \times \mathcal{O}(\mathbf{Y}))$ is computable.
\end{corollary}

\begin{lemma}
\label{lemma:translate_graphs}
$U \mapsto \{(x,n) \in \mathbf{X} \times \mathbb{N} \mid \exists V \in \mathcal{O}(\mathbb{N}) \ n \in V \wedge (x,V) \in U\} : \mathcal{O}(\mathbf{X} \times \mathcal{O}(\mathbb{N})) \to \mathcal{O}(\mathbf{X} \times \mathbb{N})$ is computable.
\begin{proof}
We start with $U \mapsto \{(V,x,n) \in \mathcal{O}(\mathbb{N} \times \mathbf{X} \times \mathbb{N} \mid n \in V \wedge (x,V) \in U\} : \mathcal{O}(\mathbf{X} \times \mathcal{O}(\mathbb{N})) \to \mathcal{O}(\mathcal{O}(\mathbb{N}) \times \mathbf{X} \times \mathbb{N})$. That this map is computable follows from open sets being effectively closed under products and intersection.

As there is a total representation $\delta_{\mathcal{O}(\mathbb{N})} : \Baire \to \mathcal{O}(\mathbb{N})$, it follows that $\mathcal{O}(\mathbb{N})$ is computably overt (\cite[Proposition 19]{pauly-synthetic-arxiv}). By \cite[Proposition 40]{pauly-synthetic-arxiv}, the existential quantifier over an overt set is a computable map from open sets to open sets, thus the claim follows.
\end{proof}
\end{lemma}

A similar characterization of computability of functions is used in effective descriptive set theory:

\begin{Mdefn}{3D}
A function $ f\: :\: \mathcal{X} \rightarrow \mathcal{Y} $ is recursive if and only if the {\em neighborhood diagram} $G^f \subseteq \mathcal{X} \times \mathbb{N}$  of $f$ defined by
\[
G^f(x,s) \Longleftrightarrow f(x) \in N(\mathcal{Y},s),
\]
is semirecursive.
\end{Mdefn}

\begin{theorem}
\label{theorem:computable_functions_the_same}
Suppose that $\mathcal{X}$, $\mathcal{Y}$ are recursively presented metric spaces, $\mathbf{X} = (M^X,d^X,(a^X_i)_{i \in \mathbb{N}})$, $\mathbf{Y} = (M^Y,d^Y,(a^Y_i)_{i \in \mathbb{N}})$ are computable metric spaces with $\mathbf{Y}$ being admissible, and $\mathbf{X}'\cong \mathbf{X}$, $\mathbf{Y}'\cong \mathbf{Y}$ are as in Theorem \ref{theorem:from_cms_to_rps}. Then:
\begin{enumerate}
\item For every $f: \mathcal{X} \to \mathcal{Y}$, $f$ is recursive exactly when $f$ is computable.
\item  For every $g: M^X \to M^Y$, $g$ is $(\mathbf{X},\mathbf{Y})$-computable \textup{(}equivalently $(\mathbf{X}',\mathbf{Y}')$-computable\textup{)} exactly when $g$ is $(\mathbf{X}',\mathbf{Y}')$-recursive.
\end{enumerate}
\end{theorem}

\begin{proof}

As before the second assertion follows from the first one and the fact that the metric spaces $\mathbf{X}'$, $\mathbf{Y}'$, are recursively presented. Regarding the first one, if $G^f$ is semirecursive then it is computably open by Theorem \ref{theorem:computably_open_sets_the_same}. By \cite[Proposition 6(7)]{pauly-synthetic-arxiv}, the map $\textrm{Cut}: \mathbf{X} \times \mathcal{O}(\mathbf{X} \times \mathbf{Y}) \to \mathcal{O}(\mathbf{Y})$ is computable. Thus, $x \mapsto \{n \in \mathbb{N} \mid (x,n) \in G^f\} : \mathcal{X} \to \mathcal{O}(\mathbb{N})$ is computable. Lemma \ref{lemma:auxiliary_to_comparison} together with Proposition \ref{prop:countablybasesubsetscott} show that we can compute $f(x)$ from $\{n \in \mathbb{N} \mid (x,n) \in G^f\}$. As the composition of computable functions is computable, we conclude that $f$ is computable.

Conversely if $f$ is computable then from Corollary \ref{corollary:computable_function} it follows that $\{(x, U) \mid f(x) \in U\} \in \mathcal{O}(\mathbf{X} \times \mathcal{O}(\mathbf{Y}))$ is computable. Using the characterization of the open sets in Proposition \ref{prop:countablybasedopens} together with Lemma \ref{lemma:auxiliary_to_comparison} shows that $\{(x, V) \mid \exists n \in V \ f(x) \in N(\mathcal{Y},n)\} \in \mathcal{O}(\mathbf{X}, \mathcal{O}(\mathbb{N}))$ is computable. Then Lemma \ref{lemma:translate_graphs} implies that $G^f$ is computably open, and so from Theorem \ref{theorem:computably_open_sets_the_same} we have that $G^f$ is semirecursive, i.e., $f$ is a recursive function.
\end{proof}

Finally we deal with points. A point $x_0$ in a computable metric space $\mathbf{X}$ is defined to be \emph{computable} if it has a computable name, i.e.~if it is the limit of a computable fast Cauchy sequence. On the other hand, point $x_1$ in a recursively presented metric space $\mathcal{X}$ is \emph{recursive} if the set $\{s \in \mathbb{N} \mid x \in N(\mathcal{X},s)\}$ is semirecursive, cf.~the comments preceding 3D.7 \cite{moschovakis}.

\begin{theorem}
\label{theorem:computable_points_the_same}
Suppose that $\mathcal{X}$ is a recursively presented metric space, $\mathbf{X} = (M,d,(a_i)_{i \in \mathbb{N}})$ is a computable metric space, and $\mathbf{X}'\cong \mathbf{X}$ is as in Theorem \ref{theorem:from_cms_to_rps}. Then:
\begin{enumerate}
\item For every $x \in \mathcal{X}$, $x$ is recursive exactly when it is computable.
\item  For every $x \in M$, $x$ is $\mathbf{X}$-computable (equivalently $\mathbf{X}'$)-computable exactly when it is $\mathbf{X}'$-recursive.
\end{enumerate}
\end{theorem}

\begin{proof}
By Proposition \ref{prop:countablybasesubsetscott}, a point in a computable metric space is computable iff $\{n \in \mathbb{N} \mid x \in \text{I}(n)\}$ is computably open. Lemma \ref{lemma:auxiliary_to_comparison} and Theorem \ref{theorem:computably_open_sets_the_same} suffice to conclude the claim.
\end{proof}

\section{On Cauchy-completions}
As a digression, we shall explore the role Cauchy-completion plays in obtained effective versions of metric spaces. An effective version of Cauchy-completion underlies both the definition of computable metric spaces and recursively presented metric spaces. A crucial distinction, though classically vacuous, lies in the question whether spaces embed into their Cauchy-completion. Our goal in this section is to explore the variations upon effective Cauchy-completion, and to subsequently understand the origin of the discrepancy exhibited in Section \ref{sec:spaces}.

Given a represented space $\mathbf{X}$ and some metric $d$ on $X$, we define the space $\mathcal{S}_C^d(\mathbf{X}) \subseteq \mathcal{C}(\mathbb{N},\mathbf{X})$ of fast Cauchy sequences by $(x_n)_{n \in \mathbb{N}} \in \mathcal{S}_C^d(\mathbf{X})$ iff $\forall i,j \geq N \ d(x_i,x_j) < 2^{-N}$. If $\mathbf{X}$ is complete, the map $\lim^d_C : \mathcal{S}_C^d(\mathbf{X}) \to \mathbf{X}$ is of natural interest (if $\mathbf{X}$ is not complete, we can still study $\lim^d_C$ as a partial map). In fact, it can characterize admissibility as follows:

\begin{proposition}
Let $\mathbf{X}$ admit a computable dense sequence. Let $d :\mathbf{X} \times \mathbf{X} \to \mathbb{R}$ be a computable metric, and let $\lim^d_C :\subseteq \mathcal{S}_C^d(\mathbf{X}) \to \mathbf{X}$ be computable. Then $\mathbf{X}$ is computably admissible.
\begin{proof}
To show that $\mathbf{X}$ is computably admissible, we need to show that $\{U \in \mathcal{O}(\mathbf{X}) \mid x \in U\} \mapsto x : \subseteq \mathcal{O}(\mathcal{O}(\mathbf{X})) \to \mathbf{X}$ is computable. We search for some point $a_1$ such that $B(a_1, 2^{-2}) \in \{U \in \mathcal{O}(\mathbf{X}) \mid x \in U\}$. Then we search for $a_2$ with $B(a_2, 2^{-3}) \in \{U \in \mathcal{O}(\mathbf{X}) \mid x \in U\}$ etc. These points form a fast Cauchy sequence converging to $x$.
\end{proof}
\end{proposition}

\begin{proposition}
\label{prop:limcomputable}
Let $\mathbf{X} = (X, \delta_\mathbf{X})$ be computably admissible and let $\dom(\delta_\mathbf{X})$ contain some computable dense sequence. Let $d :\mathbf{X} \times \mathbf{X} \to \mathbb{R}$ be a computable metric, and let $\{B(a_i,2^{-n}) \mid i,n\in \mathbb{N}\}$ be an effective countable basis for $\mathbf{X}$. Then $\lim_C^d$ is computable.
\begin{proof}
We are given some fast Cauchy sequence $(x_i)_{i \in \mathbb{N}}$ converging to some $x$ with $d(x,x_i) < 2^{-i}$ as input. As $x \in B(a_i,2^{-n}) \Rightarrow x_n \in B(a_i,2^{-n+1})$ and $x_n \in B(a_i,2^{-n}) \Rightarrow x \in B(a_i,2^{-n+1})$, we can compute $\{\langle i, n\rangle \mid x \in B(a_i,2^{-n})\} \in \mathcal{O}(\mathbb{N})$. Then we can invoke Proposition \ref{prop:countablybasesubsetscott} to extract $x$.
\end{proof}
\end{proposition}

This characterization of computable metric spaces in terms of fast Cauchy limits of course presupposes the represented space $\mathbb{R}$ with its canonical structure. In the beginnings of computable analysis, various non-standard representations of $\mathbb{R}$ have been investigated. We will investigate what happens to Cauchy completions, if some other represented space $\mathbf{R}$ (with again the reals as underlying set) is used in place of $\mathbb{R}$.

\begin{definition}
Let $\mathbf{X}$ be a represented space, such that the metric $d : \mathbf{X}\times \mathbf{X} \to \mathbf{R}$ is computable. We obtain its Cauchy-closure $\overline{\mathbf{X}}^{d,\mathbf{R}}$ by taking the usual quotient of $\mathcal{S}^d_C(\mathbf{X})$.
\end{definition}

\begin{observation}
Any computable metric space $\mathbf{X}$ embeds\footnote{A computable embedding $\mathbf{X} \hookrightarrow \mathbf{Y}$ is a computable injection $\iota : \mathbf{X} \to \mathbf{Y}$ such that the partial inverse $\iota^{-1}$ is computable, too.} into its Cauchy-closure $\overline{\mathbf{X}}^{d,\mathbb{R}}$, and $d$ can be extended canonically to $\overline{d} : \overline{\mathbf{X}}^{\overline{d},\mathbb{R}} \times \overline{\mathbf{X}}^{\overline{d},\mathbb{R}} \to \mathbb{R}$. Definition \ref{def:cms} reveals that a complete computable metric space is the Cauchy-closure of a countable metric space with continuous metric into $\mathbb{R}$.
\begin{proof}
The first part of the claim follows from Proposition \ref{prop:limcomputable} in conjunction with Proposition \ref{prop:cmshasbase}. The second part is essentially a reformulation of Observation \ref{obs:cmscompletion}. The third part is immediate from Definition \ref{def:cms}.
\end{proof}
\end{observation}

In order to find a contrasting picture of the recursively presented metric spaces, we first introduce the represented space $\mathbb{R}_{cf}$. Informally, any real number is encoded by its decimal expansion, with infinite repetitions clearly marked\footnote{For example, the unique $\mathbb{R}_{cf}$-name of $\frac{1}{3}$ is $0.\overline{3}$. The number $1$ has the names $0.\overline{9}$ and $1.\overline{0}$.}. This just ensures that $x \leq q?$ and $x \geq q?$ become both decidable for $x \in \mathbb{R}_{cf}$ and $q \in \mathbb{Q}$.

\begin{observation}
The space $\mathbb{R}_{cf}$ does not embed into $\overline{\mathbb{R}_{cf}}^{d,\mathbb{R}_{cf}}$. Let $d : \mathbf{X} \times \mathbf{X} \to \mathbb{R}_{cf}$ be a computable metric. In general, $\overline{d} : \overline{\mathbf{X}}^{d,\mathbb{R}_{cf}} \times \overline{\mathbf{X}}^{d,\mathbb{R}_{cf}} \to \mathbb{R}_{cf}$ may fail to be computable. Definition \ref{def:rpms} reveals that a recursively presented metric space is (essentially) the Cauchy-closure of a countable metric space with continuous metric into $\mathbb{R}_{cf}$.
\begin{proof}
The claims all follow from the observation that $\overline{\mathbb{R}_{cf}}^{d,\mathbb{R}_{cf}} \cong \mathbb{R} \ncong \mathbb{R}_{cf}$.
\end{proof}
\end{observation}

It is not the case, however, that the space $\mathbb{R}$ would be the only space usable in place of $\mathbf{R}$ when defining the Cauchy-closure to obtain an embedding of a space into its completion. One other example is $\mathbb{R}'$, the jump\footnote{The jump of a represented space is discussed in \cite{ziegler2,gherardi4,pauly-descriptive,pauly-descriptive-lics}.} of $\mathbb{R}$.

\section{Representations of point classes}
With a correspondence of the spaces, the continuous/computable functions and the open sets in place, we shall conclude this paper by considering higher-order classes of sets (typically called pointclasses), such as $\bolds_n^0$-sets ($n > 1$), Borel sets or analytic sets. These have traditionally received little attention in the computable analysis community, with the exception of \cite{brattka} by \name{Brattka} and \cite{selivanov5} by \name{Selivanov}. One reason for this presumably was the focus on admissible representations, i.e.~spaces carrying a topology -- and the natural representations of these classes of sets generally fail to be admissible. The ongoing development of synthetic descriptive set theory does provide representations of all the natural pointclasses.

In descriptive set theory the usual representation of pointclasses is through universal sets and good universal systems. Let $\pointcl$ be a pointclass, and $\mathbf{Z}$, $\mathbf{X}$ two spaces\footnote{Usually the spaces involved would be restricted to Polish spaces. However, the formalism is useful for us in a more general setting.}. For any $P \subseteq \mathbf{Z} \times\mathbf{X}$ and $z \in \mathbf{Z}$, we write $P_z := \{x \in \mathbf{X} \mid (z,x) \in P\}$. We write $\pointcl \upharpoonright \mathbf{X}$ for all the $\pointcl$-subsets of $\mathbf{X}$. Now we call $G \in \pointcl \upharpoonright (\mathbf{Z} \times \mathbf{X})$ a $\mathbf{Z}$-universal set for $\pointcl$ and $\mathbf{X}$ iff $\{G_z \mid z \in \mathbf{Z}\} = \pointcl \upharpoonright \mathbf{X}$.

If $\mathbf{Z} = (Z, \delta_\mathbf{Z})$ is a represented space and $G$ a $\mathbf{Z}$-universal set for $\pointcl$ and $\mathbf{X}$, then we obtain a representation $\gamma_{G}$ of $\pointcl \upharpoonright \mathbf{X}$ via $\gamma_{G}(p) = G_{\delta_\mathbf{Z}(p)}$. In this situation, we can safely assume that $\mathbf{Z} \subseteq \Baire$, and replace it by $(\dom(\delta_\mathbf{Z}),\id_{\dom(\delta_\mathbf{Z})})$ otherwise.

A $\mathbf{Z}$-universal system for $\pointcl$ is an assignment $(G^\mathbf{X})_{\mathbf{X}}$ of a $\mathbf{Z}$-universal set for $\pointcl$ and $\mathbf{X}$ for each Polish space $\mathbf{X}$. If $\mathbf{Z} = \Baire$, we suppress the explicit reference to $\mathbf{Z}$. A universal system $(G^\mathbf{X})_{\mathbf{X}}$ is \emph{good}, if for any space $\mathbf{Y}$ of the form $\mathbf{Y} = \mathbb{N}^l \times (\Baire)^k$ with $l,k \geq 0$ and any Polish space $\mathbf{X}$ there is a continuous function $S^{\mathbf{Y},\mathbf{X}} : \Baire \times \mathbf{Y} \to \Baire$ such that $(z, y, x) \in G^{\mathbf{Y} \times \mathbf{X}}\Leftrightarrow (S(z,y),x) \in G^{\mathbf{X}}$.

\newcommand{\vg}[1]{}

\emph{Comment.} In this section we consider the classical pointclasses e.g.~$\bolds^0_n$ rather than the corresponding ones $\Sigma^0_n$ in effective descriptive set theory. The classical pointclasses are also known as boldface pointclasses, because they typically arise from the effective (or else lightface) pointclasses through the process of ``boldification'', (see comments preceding 3H.1 in \cite{moschovakis}). To be more precise for every pointclass $\Gamma$ of sets in Polish spaces one defines the corresponding \emph{boldface} pointclass $\boldg$ as follows: a set $P \subseteq \mathcal{X}$, where $\mathcal{X}$ is Polish, belongs to $\boldg$ if there exists some $Q \subseteq \Baire \times \mathcal{X}$ in $\Gamma$ and some $\varepsilon \in \Baire$ such that $P$ is the $\varepsilon$-section of $Q$,
\[
P = \{x \in \mathcal{X} \mid (\varepsilon,x) \in Q\}.
\]
It is well-known that the boldface pointclasses constructed by the lightface $\Sigma^0_n, \Sigma^1_n$ are the classical Borel pointclasses $\bolds^0_n$ and $\bolds^1_n$ respectively. In fact the sets belonging to the effective pointclasses $\Gamma$ are all $\varepsilon$-sections of a set in $\Gamma$ for some recursive $\varepsilon \in \Baire$ (see 3H.1 in \cite{moschovakis}), hence the effective notion can always be recovered by the classical one\footnote{It is also worth pointing out that the effective hierarchy of lightface $\Sigma^0_n$ pointclasses can be extended transfinitely to recursive ordinals $\xi$, but it is still not known if the corresponding boldface pointclass of $\Sigma^0_\xi$ is actually the classical Borel pointclass $\bolds^0_\xi$. We nevertheless keep the latter notation for the classical Borel pointclasses with the danger of abusing the notation.}.

\begin{observation}
Let $\gamma_H$ be a representation of $\pointcl \upharpoonright \mathbf{X}$ obtained from the universal set $H$. Further let $(G^{\mathbf{X}})_{\mathbf{X}}$ be a good universal system, and let $\gamma_G$ be the induced representation of $\pointcl \upharpoonright \mathbf{X}$. Then $\id : (\pointcl \upharpoonright \mathbf{X}, \gamma_H) \to (\pointcl \upharpoonright \mathbf{X},\gamma_G)$ is continuous\footnote{This is continuity in the sense of represented spaces, generally not continuity in a topological setting.}.
\begin{proof}
By assumption, $H \in \pointcl \upharpoonright (\Baire \times \mathbf{X})$. Hence, there is some $h \in \Baire$ such that $G^{\Baire \times \mathbf{X}}_h = H$. Now $p \mapsto S^{\Baire,\mathbf{X}}(h,p)$ is a continuous realizer of $\id$.
\end{proof}
\end{observation}

As such, we see that the representations obtained from good universal system for some fixed pointclass are the weakest one (w.r.t.~continuous reducibilities) among those obtained from universal systems in general. Consequently the particular choice of a good universal system can only ever matter for computability considerations, but not for continuity.

We can now contrast the approach to representations of pointclasses via good universal system with the approach via function spaces and \Sierp-like spaces underlying \cite{pauly-descriptive,pauly-descriptive-lics}. A \Sierp-like space is a represented space $\mathbf{S}$ with underlying set $\{\top, \bot\}$ - no assumptions on the representation are made. Any such space $\mathbf{S}$ induces a pointclass $\mathcal{S}$ over the represented spaces via $U \in \mathcal{S} \upharpoonright \mathbf{X}$ iff $\chi_U : \mathbf{X} \to \mathbf{S}$ is continuous (computable), where $\chi_U(x) = \top$ iff $x \in U$. Note that this approach simultaneously provides for the effective and the classical version of $\mathcal{S}$. This pointclass comes with a represented space $\mathcal{S}(\mathbf{X})$ via the function space constructor $\mathcal{C}(-,-)$ and identification of a set and its characteristic function.

By the properties of the function space construction, we see that $\mathalpha{\ni} : \mathcal{S}(\mathbf{X}) \times \mathbf{X} \to \mathbf{S}$ is computable, which immediately implies that we may interpret $\mathalpha{\ni}$ as a $\mathcal{S}$-subset of $\mathcal{S}(\mathbf{X}) \times \mathbf{X}$. Thus, any representation of a fixed slice $\mathcal{S} \upharpoonright \mathbf{X}$ arises from some $\mathcal{S}(\mathbf{X})$-universal set. By moving along the representation, we may replace $\mathcal{S}(\mathbf{X})$ with some suitable $\mathbf{Z} \subseteq \Baire$ here.

Next, we may relax the requirements for $\mathbf{Z}$-universal systems for $\pointcl$ to allow $\mathbf{Z}$ to vary as $\mathbf{Z}_\mathbf{X}$
with the space $\mathbf{X}$, and will also let $\mathbf{X}$ range over all represented spaces, rather than just Polish spaces. The resulting notion shall be called a generalized universal system. Such a system $(\mathbf{Z}_\mathbf{X}, G^{\mathbf{X}})_{\mathbf{X}}$ is good, if for any represented spaces $\mathbf{Y}$, $\mathbf{X}$ there is a continuous function $S^{\mathbf{Y},\mathbf{X}} : \mathbf{Z}_{\mathbf{X} \times \mathbf{Y}} \times \mathbf{Y} \to \mathbf{Z}_\mathbf{X}$ such that $(z, y, x) \in G^{\mathbf{Y} \times \mathbf{X}}\Leftrightarrow (S(z,y),x) \in G^{\mathbf{X}}$.

\begin{observation}
Let the generalized universal system $(\mathbf{Z}_\mathbf{X}, G^{\mathbf{X}})_{\mathbf{X}}$ be obtained from the \Sierp-like space $\mathbf{S}$. Then it is good.
\begin{proof}
$S^{\mathbf{Y},\mathbf{X}} : \mathcal{C}(\mathbf{X} \times \mathbf{Y}, \mathbf{S}) \times \mathbf{Y} \to \mathcal{C}(\mathbf{X},\mathbf{S})$ is realized via partial function application.
\end{proof}
\end{observation}

For most natural choices of a \Sierp-like space $\mathbf{S}$, we may actually replace the occurrence of $\mathcal{C}(\mathbf{X},\mathbf{S})$ in the induced generalized universal system by $\Baire$ again, thus closing the distance between the two approaches. We recall from \cite{kreitz} that a representation $\delta : \subseteq \Baire \to \mathbf{X}$ is called precomplete, if for any computable partial $F : \subseteq \Baire \to \Baire$ there is a computable total $\overline{F} : \Baire \to \Baire$ such that $\delta \circ F(p) = \delta \circ \overline{F}(p)$ for all $p \in \dom(\delta \circ F(p))$. Now note that if $\mathbf{S}$ admits a precomplete representation, then $\mathcal{C}(\mathbf{X}, \mathbf{S})$ admits a total representation for any $\mathbf{X}$. Subsequently, we note:

\begin{observation}
Let the \Sierp-like space $\mathbf{S}$ admit a precomplete representation. Then it induces a pointclass $\mathcal{S}$ together with a good universal system.
\end{observation}

We will now explore which pointclasses on Polish spaces are obtainable from some \Sierp-like space. First, note that any such class $\mathcal{S}$ is closed under taking preimages under continuous functions. Then, for any Polish space $\mathbf{X}$ and total representation $\delta : \Baire \to \mathbf{X}$, we observe that $A \in \mathcal{S} \upharpoonright \mathbf{X}$ iff $\delta^{-1}(A) \in \mathcal{S} \upharpoonright \Baire$. Generally, we shall call any pointclass satisfying this property for all total admissible representations of Polish spaces to be \emph{$\Baire$-determined}.

\begin{proposition}
Let $\pointcl$ be $\Baire$-determined, closed under continuous preimages and admit a good universal system. Then there is some \Sierp-like space $\mathbf{S}$ with $\pointcl = \mathcal{S}$.
\begin{proof}
Let $G \subseteq \Baire \times \Baire$ be a universal set for $\pointcl$ and $\Baire$. We define a representation $\delta_G : \Baire \to \{\top,\bot\}$ via $\delta_G(\langle p,q\rangle) = \top$ iff $(p, q) \in G$. Let the resulting space be $\mathbf{S}$. We claim that the pointclass induced by $\mathbf{S}$ coincides with $\pointcl$.

Let $A \in \pointcl \upharpoonright \mathbf{X}$. Then $\delta_\mathbf{X}^{-1}(A) \in \pointcl \upharpoonright \Baire$. Thus, there is some $a \in \Baire$ with $q \in A \Leftrightarrow (a, q) \in G$. Now $q \mapsto \langle a, q\rangle$ is a continuous realizer of $\chi_A : \mathbf{X} \to \mathbf{S}$.

Conversely, assume $\chi_A \in \mathcal{C}(\mathbf{X},\mathbf{S})$. Let $c_A : \Baire \to \Baire$ be a continuous realizer of $\chi_A$. Note $\left ((\pi_1, \pi_2) \circ c_A\right )^{-1}(G) = \delta_\mathbf{X}^{-1}(A)$. The left hand side of this equation shows that the set is in $\pointcl$, as $\pointcl$ is closed under continuous preimages. The right hand side then implies that $A \in \pointcl \upharpoonright \mathbf{X}$, as $\pointcl$ is $\Baire$-determined.
\end{proof}
\end{proposition}

\begin{proposition}
\label{prop:bairedetermined}
Let $\pointcl$ be a pointclass.
\begin{enumerate}
\item If $\pointcl$ is $\Baire$-determined, then so are $\pointcl^\mathfrak{C}$ and $\pointcl \cap \pointcl^\mathfrak{C}$; where $\pointcl^\mathfrak{C} := \{A^C \mid A \in \pointcl\}$.
\item For countable ordinals $\alpha$, $\bolds^0_\alpha$ is $\Baire$-determined.
\item $\bolds^1_n$ is $\Baire$-determined, $n \geq 1$.
\end{enumerate}
\begin{proof}
\begin{enumerate}
\item Just observe that $\delta^{-1}(A^C) = (\delta^{-1}(A))^C$.
\item This is a result by \name{Saint Raymond} \cite{SaintRaymond} (cf.~\cite{debrecht4})
\item Let $\mathbf{X}$ be a Polish space and $A \subseteq \mathbf{X}$ be $\bolds^1_n$. Since the continuous preimage of a $\bolds^1_n$ set is $\bolds^1_n$, and $\delta$ is continuous it follows that $\delta^{-1}[A] \in \bolds^1_n \upharpoonright \Baire$. Conversely using that as a representation, $\delta$ is surjective, we have that $A = \delta[\delta^{-1}[A]]$. So if $\delta^{-1}[A]$ is a $\bolds^1_n$ subset of $\Baire$ it follows from the closure of $\bolds^1_n$ under continuous images that $A \in \bolds^1_n \upharpoonright \mathbf{X}$.
\end{enumerate}
\end{proof}
\end{proposition}

{\bf Conclusion:} The approaches to continuity and computability for $\bolds^0_\alpha$ and $\bolds^1_1$ from effective descriptive set theory and synthetic descriptive set theory coincide.

\quad

A very important pointclass not yet proven to receive equivalent treatment are the Borel sets $\BC$, alternatively $\boldd_1^1$ by Suslin's theorem (e.g.~\cite{moschovakis2}). There cannot be any $\Baire$-universal Borel sets\footnote{Any such set would fall into $\bolds^0_\alpha$ for some countable ordinal $\alpha$, but then cannot have any set $A \in \bolds_{\alpha+1}^0 \setminus \bolds_{\alpha}^0$ as a section.} -- however, there are $\mathbf{B}$-universal sets for $\boldd_1^1$, with non-Polish $\mathbf{B}$. Such a set can be obtained from the Borel codes used in effective descriptive set theory. We currently cannot prove uniform equivalence of the two approaches for Borel sets on arbitrary Polish spaces, as this would require a uniform version of \name{Saint Raymond}'s result in \cite{SaintRaymond}\footnote{For our purposes, this result is that if $\delta$ is a standard represention of a computable Polish space $\mathbf{X}$, and $A$ is a Borel subset of Baire, then $\delta[A]$ is a Borel subset of $\mathbf{X}$. A uniform version would allow us to compute a Borel code for $\delta[A]$ from a Borel code for $A$.}. Thus, we first provide a non-uniform treatment of Borel sets on arbitrary Polish spaces, and then a uniform treatment of Borel subsets of $\Baire$.

\begin{definition}(\cite{moschovakis} 3H)
The set of \emph{Borel codes} $\bcode \subseteq \Baire$ is defined by recursion as follows
\begin{align*}
p \in \bcode_0 &\iff p(0) = 0\\
p \in \bcode_\alpha &\iff p = 1\langle p_0, p_1, \ldots, \rangle \ \& \ (\forall n)(\exists \beta < \alpha)[p_n \in \bcode_\beta]\\
\bcode = \cup_{\alpha} \bcode_\alpha& \hspace{2mm} \textrm{for all countable ordinals $\alpha$.}
\end{align*}
\end{definition}
With an easy induction one can see that $\bcode_\alpha \subseteq \bcode_\beta$ for all $\alpha < \beta$ and that $\bcode_\alpha$ is a Borel set.

For all $p \in \bcode$ we denote by $|p|$ the least ordinal $\alpha$ such that $p \in \bcode_\alpha$ (\footnote{This essentially provides a representation of the space of all countable ordinals. This idea is investigated in some detail in \cite{pauly-ordinals-arxiv,pauly-ordinals-mfcs}.}). It is not hard to verify that
\[
|1\langle p_0,p_1,\ldots\rangle| = \sup_{n \in \mathbb{N}} |p_n| + 1
\]
Let $\mathbf{X}$ be a Polish space, and $\delta_\mathcal{O} : \Baire \to \mathcal{O}(\mathbf{X})$ a standard representation of its open sets. For some subset $A \subseteq \mathbf{X}$, let $A^C$ denote its complement $\mathbf{X} \setminus A$. For all countable ordinals $\alpha$ we define the function $\bcodefun_\alpha^{\mathbf{X}}: \bcode_\alpha \to \BC \upharpoonright \mathbf{X}$ recursively by
\begin{align*}
\bcodefun_0^\mathbf{X}(0p) =& \delta_\mathcal{O}(p)\\
\bcodefun_\alpha^\mathbf{X}(1\langle p_0, p_1, \ldots\rangle) =& \bcodefun_{|\bigcup_{n}p_n|}^{\mathbf{X}}\left(\bigcup_{n}p_n\right)^C.
\end{align*}
An easy induction shows that the function $\bcodefun_\alpha^\mathbf{X}$ is onto $\bolds^0_\alpha \upharpoonright \mathbf{X}$, and that $\bcodefun_\beta^\mathbf{X} \upharpoonright \bcode_\alpha = \bcodefun_\alpha^\mathbf{X}$ for all $\alpha < \beta$. So one can define the \emph{Borel coding} $\bcodefun^{\mathbf{X}}: \bcode \to \BC \upharpoonright \mathbf{X}$ by
\[
\bcodefun^{\mathbf{X}}(p) = \bcodefun_{|p|}^{\mathbf{X}}(p).
\]
so that the family $\bolds^0_\alpha \upharpoonright \mathbf{X}$ is exactly the family of all $\bcodefun^{\mathbf{X}}(p)$ for $p \in \bcode_\alpha$, in particular a set $A \subseteq \mathbf{X}$ is Borel exactly when $A = \bcodefun^{\mathbf{X}}(p)$ for some $p \in \bcode$.

The following are more or less well-known facts in descriptive set theory:

\begin{lemma}
\label{lemma:dstfacts}
\begin{enumerate}
\item For all countable ordinals $\alpha$ the set $\{p \in \bcode \mid |p| \leq \alpha\}$ is Borel.
\begin{proof}
This is because $\{p \in \bcode \mid |p| \leq \alpha\} = \bcode_\alpha$.
\end{proof}
\item The set $\bcode$ is a $\Pi^1_1$ subset of $\Baire$ and so in particular it is a $\boldp^1_1$ set.
\begin{proof}
The latter is a consequence of 7C.8 in \cite{moschovakis}, since one can see that the set $\bcode$ is the least fixed point of a suitably chosen monotone operation.\smallskip
\end{proof}
\item There exists a $\bolds^1_1$ relation $\leq_{\Sigma} \subseteq \Baire \times \Baire$ such that for all $p \in \bcode$ and all $q \in \Baire$ we have that
\[
[q \in \bcode \ \& \ |q| \leq |p|] \iff q \leq_{\Sigma}  p.\footnote{The idea behind this condition can be found in the notion of $\Gamma$-norms, see \cite{moschovakis} 4B. In the same way, one could also obtain a $\boldp^1_1$-relation with the same property (note that it does not follow that there is a $\boldd_1^1$-relation).}\]
\begin{proof}
Note that $|1\langle q_0, q_1, \ldots\rangle| \leq |1\langle p_0, p_1,\ldots\rangle|$ iff $\exists t \in \Baire$ s.t.~$\forall n \in \mathbb{N} \ |q_n| \leq |p_{t(n)}|$, assuming $q_i, p_i \in \bcode$. Building upon this idea, consider the closed relation $R$ defined as the least fixed point of:
\[R(p,q,\langle t',\langle t_0, t_1, \ldots \rangle\rangle) :\Leftrightarrow q(0) = 0 \vee \left (p = 1\langle p_0,p_1,\ldots\rangle \wedge q = 1\langle q_0,q_1,\ldots\rangle \wedge \forall n \in \mathbb{N} \ R(p_n,q_{t'(n)},t_n)\right )\]
Now $q \leq_\Sigma p :\Leftrightarrow \exists t \in \Baire \ R(p,q,t)$ is a $\bolds^1_1$ relation, and satisfies our criterion.
\end{proof}

\item The set $\bcode$ is not a Borel subset of $\Baire$.
\begin{proof}
We will show that if $\bcode$ were Borel, then the set of well-founded trees would be analytic, which is a contradiction (as shown e.g.~in \cite[Section 11.8]{bruckner}).

Note that a tree\footnote{Here we understand a tree to be a subset of $\mathbb{N}^*$ that is closed under taking prefixes.} $T$ is well-founded iff there exists an assignment $P: T \to \bcode$ such that
for all $u,v \in T$ if $v$ extends $u$ then $|P(v)| < |P(u)|$.

This is easy to see: If $T$ is well-founded then we use bar recursion to get
$P$ such that $|P(u)| = \sup |(un)| + 1$. Conversely if $P$ is such an
assignment and $T$ contained an infinite branch then we would get a strictly
decreasing sequence of ordinals, a contradiction.

Now condition $|P(v)| < |P(u)|$ can be replaced by $S(P(v)) \leq_\Sigma P(u)$,
with $\leq_\Sigma$ as above, and $S$ is a continuous
function such that $|S(q)| = |q|+1$. Thus, we have
\begin{align*}
&T \textnormal{ is well-founded} \ \Leftrightarrow\\
& \hspace*{5mm} \exists P \ \ \forall u, v \in \mathbb{N}^*. u \in T \Rightarrow P(u) \in \bcode \ \text{and} \ v \text{ extends } u \Rightarrow \ S(P(v)) \leq_\Sigma P(u).
\end{align*}
The preceding $P$ varies through the set of all functions from $\mathbb{N}^*$ to $\Baire$, and the latter set is homeomorphic to $\Baire$. If the set $\bcode$ were Borel, then the right-hand side of the preceding equivalence would define a $\bolds^1_1$ set, and hence the set of all well-founded trees would be $\bolds^1_1$, a contradiction.
\end{proof}

\item Let $f : \Baire \to \bcode$ be Borel measurable. Then there is a countable ordinal $\alpha$ such that $\forall p \in \Baire \ |f(p)| \leq \alpha$.
\begin{proof}
If this were not the case we would have that
\[
q \in \bcode \iff (\exists p)[q \leq_\Sigma f(p)],
\]
where $\leq_{\Sigma}$ is as above. Since $f$ is Borel measurable the preceding equivalence would imply that the set $\bcode$ is a $\bolds^1_1$ subset of $\Baire$. Hence from the Suslin Theorem it would follow that $\bcode$ is a Borel set, a contradiction and our claim is proved.
\end{proof}
\end{enumerate}
\end{lemma}

\begin{definition}
We define the \Sierp -like space $\mathbf{S}_{\BC} = (\{\bot, \top\}, \delta_\BC)$ recursively via
\begin{align*}
\delta_\BC(p) \ \textrm{is defined}& \ \iff \ p \in \bcode\\
\delta_\BC(0p) = & \delta_\mathbb{S}(p)\\
\delta_\BC(1\langle p_0, p_1, \ldots\rangle) = & \bigvee_{i \in \mathbb{N}} \neg \delta_\BC(p_i).
\end{align*}
\end{definition}

Note that by construction of $\mathbf{S}_\BC$, we find that $\mathalpha{\in} : \Baire \times \BC \to \mathbf{S}_\BC$ is computable.

\begin{proposition}
\label{prop:sbworks}
Fix a Polish space $\mathbf{X}$. For $A \subseteq \mathbf{X}$ we find the following to be equivalent:
\begin{enumerate}
\item $A \in \BC \upharpoonright \mathbf{X}$
\item $\chi_A : \mathbf{X} \to \mathbb{S}_\BC$ is continuous.
\item $\chi_A : \mathbf{X} \to \mathbb{S}_\BC$ is Borel measurable.
\end{enumerate}
\begin{proof}
\begin{description}
\item[$1. \Rightarrow 2.$]
Fix a total admissible representation $\delta_\mathbf{X} : \Baire \to \mathbf{X}$. Let us assume that $A \in \BC \upharpoonright \mathbf{X}$. Then $\delta_\mathbf{X}^{-1}(A) \in \BC \upharpoonright \Baire$. If $a$ is a Borel code for $\delta_\mathbf{X}^{-1}(A)$, then $q \mapsto \mathalpha{\in}(q,a)$ is a continuous realizer for $\chi_A : \mathbf{X} \to \mathbb{S}_\BC$.
\item[$2. \Rightarrow 3.$] Trivial.
\item[$3. \Rightarrow 1.$]
Now let us assume that $\chi_A : \mathbf{X} \to \mathbf{S}_\BC$ is Borel measurable. Let $c_A : \Baire \to \Baire$ be a Borel measurable realizer of $\chi_A$. We remark that $c_A(p) \in \bcode$ for all $p \in \Baire$. Consider now some countable ordinal $\alpha_A$ such that $|c_A(p)| < \alpha_A$ for all $p \in \Baire$, which we may obtain from Lemma \ref{lemma:dstfacts} (5). The set $S_{\alpha_A} := \{p \in \Baire \mid \delta_{\BC}(p) = \top \wedge |p| \leq \alpha_A\}$ is a Borel subset of $\Baire$. Then $c_A^{-1}(S_{\alpha_A}) = \delta_\mathbf{X}^{-1}(A)$ is Borel as well and hence it is $\bolds^0_{\beta_A}$ for some countable ordinal $\beta_A$.  By Proposition \ref{prop:bairedetermined} (2), we find that $A \in \bolds_{\beta_A}^0 \upharpoonright \mathbf{X}$, in particular, $A$ is Borel.
\end{description}
\end{proof}
\end{proposition}

As announced above, we will proceed to show that for Baire space the representation of $\BC$ via Borel codes is computably equivalent to the representation via the function space into $\mathbf{S}_\BC$. In this, we will consider the Borel codes to be the default representation of $\BC \upharpoonright \Baire$. A new ingredient of the proof will be:
\begin{lemma}
\label{lemma:r}
The operation $r : \subseteq \BC \upharpoonright \Baire \to \BC \upharpoonright \Baire$ with $\dom(r) = \{A \in \BC \upharpoonright \Baire \mid A \subseteq \bcode\}$ and $r(A) = \{p \in A \mid \delta_\BC(p) = \top\}$ is well-defined and computable.
\begin{proof}
We start by providing $\Sigma_1^1$-sets $T$ and $B$, such that $\delta_\BC(p) = \top \Leftrightarrow p \in \bcode \cap T$ and $\delta_\BC(p) = \bot \Leftrightarrow p \in \bcode \cap B$. This is done by constructing two $\Pi^0_1$-sets $P, Q \subseteq \Baire \times \Baire$ via an interleaving fixed point construction\footnote{The reader coming from a computable analysis background may prefer to see the following as instructions for a dove-tailing programme trying to disprove $(p,q) \in P$ or $(p,q) \in Q$ by unraveling the instructions. If ever one of the first two cases is reached and yields a negative answer, this is propagated back and disproves the original membership query. It is perfectly fine to have ill-founded computation paths, these can never yield contradictions and thus may cause queries to fall in $P$ or $Q$ where the first parameter is not a Borel code.}:
\[\begin{array}{rcl}
(0p, q) \in P & :\Leftrightarrow & p = 0^\mathbb{N}\\
(0p, nq) \in Q & :\Leftrightarrow & p(n) = 1\\
(1\langle p_0, p_1, \ldots\rangle, \langle q_0, q_1, \ldots\rangle) \in P & :\Leftrightarrow & \forall n \in \mathbb{N} \ (p_n, q_n) \in Q\\
(1\langle p_0, p_1, \ldots \rangle, nq) \in Q & :\Leftrightarrow & (p_n,q) \in P\\
\end{array}\]
Now $p \in T :\Leftrightarrow \exists q \in \Baire \ (p,q) \in Q$ and $p \in B : \Leftrightarrow \exists q \in \Baire \ (p,q) \in P$ are our desired sets.

Given $A \in \BC \upharpoonright \Baire$ we can compute $A \cap T \in \bolds_1^1 \upharpoonright \Baire$ and $A^C \cup B \in \bolds_1^1 \upharpoonright \Baire$, and note that $A \subseteq \bcode$ implies $(A \cap T)^C = A^C \cup B$, so by applying the effective Suslin theorem (\name{Moschovakis} \cite{moschovakis2}) we can obtain $r(A) = A \cap T \in \BC$.
\end{proof}
\end{lemma}

\begin{theorem}
\label{theo:boreleffectivecoincidence}
The map $A \mapsto \chi_A : \BC \upharpoonright \Baire \to \mathcal{C}(\Baire, \mathbf{S}_\BC)$ is a computable isomorphism.
\begin{proof}
That this map is computable follows by currying from the computability of $\mathalpha{\in} : \Baire \times \BC \to \mathbf{S}_\BC$; that it is a bijection from Proposition \ref{prop:sbworks}. It only remains to prove that its inverse is computable, too.

Given $\chi_A \in \mathcal{C}(\Baire, \mathbf{S}_\BC)$, we can compute the $\bolds_1^1$ set $\chi_A[\Baire]$. Then we use the effective Suslin theorem (\name{Moschovakis} \cite{moschovakis2}) on $\chi_A[\Baire]$ and the $\Sigma_1^1$-set $\bcode^C$ to obtain some $B \in \BC \upharpoonright \Baire$ with $\chi_A[\Baire] \subseteq B \subsetneq \bcode$. Using the computable map $r$ from Lemma \ref{lemma:r} we can then obtain $A \in \BC \upharpoonright \Baire$ as $A = \chi_A^{-1}(r[B])$.
\end{proof}
\end{theorem}

Theorem \ref{theo:boreleffectivecoincidence} allows us to conclude some effective closure properties of $\BC$ either directly, or using some basic properties of $\mathbf{S}_\BC$. We start with the latter:

\begin{proposition}
\label{prop:sbcclosure}
The following maps are computable:
\begin{enumerate}
\item $\neg : \mathbf{S}_\BC \to \mathbf{S}_\BC$
\item $\wedge, \vee : \mathbf{S}_\BC \times \mathbf{S}_BC \to \mathbf{S}_\BC$
\item $\bigwedge, \bigvee : \mathcal{C}(\mathbb{N},\mathbf{S}_\BC) \to \mathbf{S}_\BC$
\end{enumerate}
\begin{proof}
\begin{enumerate}
\item This is realized by $p \mapsto \langle 1, p, p, p, \ldots\rangle$.
\item This follows from (3.).
\item $(p_i)_{i \in \mathbb{N}} \mapsto \langle 1, \langle 1, p_0, p_1, \ldots\rangle, \langle 1, p_0, p_1, \ldots\rangle, \ldots\rangle$ realizes $\bigwedge$. Computability of $\bigvee$ follows using de Morgan's law and $(1.)$.
\end{enumerate}
\end{proof}
\end{proposition}

\begin{corollary}
\label{corr:computableonbc}
The following maps are computable:
\begin{enumerate}
\item $(f, U) \mapsto f^{-1}(U) : \mathcal{C}(\Baire, \Baire) \times \BC \upharpoonright \Baire \to \BC \upharpoonright \Baire$
\item $U \mapsto U^C : \BC\upharpoonright \Baire \to \BC\upharpoonright \Baire$
\item $\cap, \cup : \BC\upharpoonright \Baire \times \BC\upharpoonright \Baire \to \BC\upharpoonright \Baire$
\item $\bigcap, \bigcup : \mathcal{C}(\mathbb{N},\BC\upharpoonright \Baire) \to \BC\upharpoonright \Baire$
\end{enumerate}
\begin{proof}
Use the characterization of $\BC$ given by Theorem \ref{theo:boreleffectivecoincidence}. The first map is realized by function composition, the remaining by composing with the appropriate map from Proposition \ref{prop:sbcclosure}.
\end{proof}
\end{corollary}

While proving the results of Corollary \ref{corr:computableonbc} directly would not have been particularly cumbersome either, the present approach immediate generalizes to all represented spaces. In analogy to Theorem \ref{theo:boreleffectivecoincidence}, we could \emph{define} the space $\BC_\mathbf{X}$ of Borel subsets of some represented space $\mathbf{X}$ by identifying $U \subseteq \mathbf{X}$ with (continuous) $\chi_U : \mathbf{X} \to \mathbf{S}_\BC$. Corollary \ref{corr:computableonbc} then immediately shows that $\BC_\mathbf{X}$ has the expected effective closure properties.
\section{Conclusions}
We have demonstrated that the computability notions used in computable analysis (and synthetic descriptive set theory) and effective descriptive set theory respectively coincide for objects in the scope of both. When it comes to metric spaces, the scope of effective descriptive set theory is more restrictive, however, the difference disappears modulo a rescaling of the metric. While the requirements for pointclasses to be treatable in the two frameworks differ significantly, the computability notions for $\bolds_\alpha^0$, $\boldp_\alpha^0$, $\bolds^1_1$ and $\boldp^1_1$ coincide for Polish spaces, and $\BC$ ($\boldd^1_1$) is the same in both frameworks for Baire space.

\bibliographystyle{eptcs}
\bibliography{metricspaces}

\section*{Acknowledgements}
The work presented here benefited from the Royal Society International Exchange Grant IE111233. The second author was supported by an EPSRC-funded UROP bursary. The first author would like to thank Yiannis Moschovakis for valuable discussions, and is grateful to Ulrich Kohlenbach for his continuing substantial support.
\end{document}